\documentclass[12pt,reqno]{amsart}
\usepackage[samelinetheorem,notitle]{maherart}

\hypersetup{
    pdftitle =
        {Queue Design for Staking and Restaking},
    pdfauthor =
        {Michael Neuder, Mallesh M Pai and Max Resnick},
    pdfsubject=
        {Queue Design}
}

\usepackage{tabularx}
\usepackage{multirow}
\usepackage{longtable}

    \newcolumntype{L}{>{\raggedright\arraybackslash}X}
\newcommand{\mech}{\mathcal{M}}
\usepackage{xcolor}

\usepackage{nicematrix,booktabs,caption}

\begin{document}
\raggedbottom
\author[Neuder]{Michael Neuder$^\textrm{A}$}
\address{$^\textrm{A}$ Ethereum Foundation}

\author[Pai]{Mallesh M Pai$^\textrm{B}$}
\address{$^\textrm{B}$ Special Mechanisms Group and Rice University}

\author[Resnick]{Max Resnick$^\textrm{C}$}
\address{$^\textrm{C}$ Special Mechanisms Group}

\title{Optimizing Exit Queues for Proof-of-Stake Blockchains: A Mechanism Design Approach}
\thanks{\textit{The authors thank Aditya Asgaonkar, Vitalik Buterin, Francesco D'Amato, Barnabé Monnot,  and  Tim Roughgarden for helpful discussions. Code available at \href{https://github.com/michaelneuder/withdrawals/}{https://github.com/michaelneuder/withdrawals/}.}}

\begin{abstract}
Byzantine fault-tolerant consensus protocols have provable safety and liveness properties for static validator sets. In practice, however, the validator set changes over time, potentially eroding the protocol’s security guarantees. For example, systems with accountable safety may lose some of that accountability over time as adversarial validators exit. As a result, protocols must rate limit entry and exit so that the validator set changes slowly enough to ensure security. Here, the system designer faces a fundamental trade-off. The harder it is to exit the system, the less attractive staking becomes; the easier it is to exit the system, the less secure the protocol will be.

This paper provides the first systematic study of exit queues for Proof-of-Stake blockchains. Given a collection of validator-set consistency constraints imposed by the protocol, the social planner’s goal is to provide a constrained-optimal mechanism that minimizes disutility for the participants. We introduce the \texttt{MINSLACK} mechanism, a dynamic capacity first-come-first-served queue in which the amount of stake that can exit in a period depends on the number of previous exits and the consistency constraints. We show that \texttt{MINSLACK} is optimal when stakers are homogeneous in their need for liquidity. When stakers are heterogeneous, the optimal mechanism resembles a priority queue with dynamic capacity. However, this mechanism must reserve exit capacity for the future in case a staker with a more urgent need for liquidity arrives. We conclude with a survey of known consistency constraints and highlight the diversity of existing exit mechanisms.
\end{abstract}

\begin{titlepage}
    \maketitle
\end{titlepage}

\section{Introduction}
In Proof-of-Stake networks, validators must stake tokens to participate in the consensus protocol. These staked tokens serve two purposes. First, they solve the problem of Sybil resistance: agents who operate two validators must procure twice as much stake as those who only manage one. Second, they allow the protocol to hold validators accountable for violating the predefined rules. A validator's stake is held in escrow and can be \textit{slashed} if adversarial behavior is detected, providing \emph{crypto-economic security} to the system.\footnote{See \cite{deb2024stakesure} for an extended definition of crypto-economic security.} Most modern blockchains (e.g., Ethereum and Solana) implement a version of Proof-of-Stake and the principles of staking have been extended beyond base-layer chains and into the smart contract layer (e.g., re-staking as popularized by EigenLayer). 

The literature typically treats the set of stakers as static to establish positive results; however, in practice, staking protocols have a validator set that changes over time. New agents may arrive and wish to stake, while existing stakers may want to withdraw their tokens for use elsewhere (see, e.g., \cite{chitra2021competitive}). But how should a stake-based protocol design this egress procedure?\footnote{Similar considerations apply to the design of deposit (ingress) procedures; we focus on withdrawals (egress) in the present paper.} 
There are two competing desiderata. The first is the security of the underlying protocol. For example, if a malicious validator can corrupt the service for personal gain but then withdraw their stake before the corruption is detected, the validator is immune to punishment, and the protocol is not secure. We describe these concerns in more detail in Section \ref{sec:limit}. 
The second is ensuring that validators can quickly enter and exit the system since delays decrease the utility of participation. Offering fast withdrawals also indirectly benefits the protocol since, \emph{ceteris paribus}, a more rigid protocol must offer higher rewards in the form of emissions to compensate users for the decrease in their optionality. While these general principles are well understood, the design of the exit procedures in the context of blockchains has not yet attracted much formal attention.\footnote{We defer a discussion of the related literature to Section \ref{sec:rellit}.} 
As a consequence, the optimal queue designs we suggest in Sections \ref{sec:homogeneous} and \ref{sec:heterogenous} perform much better than those currently used in practice, which we survey in Section \ref{sec:theory-and-practice}. 

The first contribution of this paper is to formally define the designer’s dilemma as a constrained optimization problem: minimizing the adverse effects of withdrawal delays while satisfying the protocol’s safety constraints. In the setting where all validators have the same time sensitivity, we show that a stateful, first-come-first-served queue where the amount of stake withdrawn in each period depends on the history of previous periods is constrained optimal.

However, even among honest validators, the desire to exit can be heterogeneous---for example, a capital-constrained validator might need to withdraw urgently to meet a margin call elsewhere. In this setting, a first-come-first-served queue is no longer optimal, because time-sensitive stakers incur a higher penalty from delay, so allowing them to pay to cut the line increases aggregate welfare. Therefore, the optimal mechanism processes transactions based on their willingness to pay for priority. Further, in some cases, the optimal mechanism reserves capacity for the future in case more time-sensitive agents arrive. We formally define this mechanism as the solution to a Markov Decision Process (MDP) and show that an appropriately defined dynamic Vickrey-Clark-Groves (VCG) mechanism can implement the efficient outcome.  

We complement these results with a survey to connect our theoretical model to practice. First, we discuss the exit mechanisms in use today by popular blockchains. Our results suggest that some of these mechanisms are either (highly) sub-optimal or the designers believed the mechanism should satisfy additional constraints external to our model. Further, we should note that no protocol that we are aware of uses a payment-for-priority mechanism.\footnote{Priority payment is standard in other congested parts of blockchains, notably in the context of transaction fees, and a substantial literature explores the design of such fees, see, e.g., \cite{roughgarden2020transaction, huberman2021monopoly}.}  Combined with our theoretical results, this collection may be helpful for blockchains and staking protocols more generally to design or improve their exit procedures. 

\textbf{Organization.} Section~\ref{sec:rellit} presents the related literature. Section~\ref{sec:model} defines the model and outlines the form of the security constraints of a staking system. Section~\ref{sec:homogeneous} studies the common-value setting by defining \texttt{MINSLACK} and proving its optimality. Section~\ref{sec:heterogenous} introduces heterogeneous values, formalizes the extended problem as an MDP, and presents numerical results quantifying the performance of various algorithms. Section~\ref{sec:theory-and-practice} justifies the form of the constraints, presents the intricacies of the Ethereum design, and surveys other staking withdrawal procedures. Section~\ref{sec:concl} concludes.

\section{Literature Review}\label{sec:rellit}
Early in the development of Ethereum, Vitalik outlined concerns about a long-range attack on a Proof-of-Stake blockchain \citep{buterin2014howilearned}. In particular, he described the potential ``nothing-at-stake'' attack, in which a malicious staker withdraws his \texttt{ETH} while building a competing fork starting from a historical epoch before he withdrew. This way, the staker cannot be punished for creating the fork because he has exited from the consensus mechanism. 
The solution, he argued, was \emph{weak subjectivity}, where nodes locally store a subjectivity ``checkpoint'' block and ignore any messages from before that epoch \citep{aditya2020wseth}. Weak-subjectivity checkpoints prevent long-range attacks but require the validator set to change slowly enough to reach a subjectivity checkpoint without a nothing-at-stake attack. The naïve approach simply delays \textit{all} withdrawals for the weak-subjectivity period, guaranteeing the chain's safety. Buterin argued in \cite{buterin2018averagecase} that this imposed too strict a penalty on validators who wanted to withdraw under normal circumstances when there was no evidence of an attack, arguing for an exit queue model instead. 
The simple first-come-first-served (FCFS) exit queue proposed allowed validators to ``withdraw very quickly in `the normal case' but not during attacks''. \cite{buterin2019entriesnotwithdrawals} elaborated on this point and fleshed out the arguments for an FCFS exit queue with a fixed number of exits per epoch (as a fraction of the total stake). 
\cite{buterin2019weaksubjectivequeue} gave a formal case for why the consistency imposed by the exit queue was enough to safely last until the next weak-subjectivity checkpoint through an inductive argument. As detailed in Section~\ref{sec:theory-and-practice}, the FCFS exit queue has been used in Ethereum since April 2023, when the Shanghai/Capella hard fork enabled beacon chain validators to withdraw. More generally, Buterin was concerned with preserving the formal property known as \textit{accountable safety} \citep{buterin2017casper}. Accountable safety guarantees that in the case of a safety violation, the fault is attributable to a subset of validators (because they must have signed conflicting attestations).

\cite{lewispye2024permissionless,budish2024economic} formalized the economic limits of consensus mechanisms and showed that no partially synchronous protocol can fully implement slashing rules without bounding the resolution time of communication between honest nodes. Given some bound on this overhead, protocols could implement slashing against an attacker with $<2/3$ of the total stake, a positive result that justifies using the weak-subjectivity period as a heuristic for preventing nothing-at-stake attacks. \cite{neu2023accountablesafety} formalized the relationship between accountable safety and finality, while \cite{asgaonkar2024confirmation} introduced a new confirmation rule for potentially improving the pre-finality guarantees for transactions in Ethereum. \cite{oisin} proposed allowing some withdrawals to be processed ahead of others by the nature of originating from a different source that required payment.

Systematic attention to the design of exit procedures in blockchains has been sparse; however, mechanism design has proved useful for blockchain designers in other contexts, particularly in designing transaction-fee mechanisms. The question here is similar: if there is a finite supply of block space and demand may exceed supply, how should the block space be allocated? Bitcoin used a simple `pay-as-bid' mechanism, which was fruitfully studied using tools from queueing theory in \cite{huberman2021monopoly}. Pay-as-bid mechanisms result in strategic bidding, contributing to poor user experience. In 2021, Ethereum adopted a dynamic reserve price mechanism, EIP-1559, which was comprehensively studied in the seminal \cite{roughgarden2020transaction} (see also \cite{roughgarden2021transaction}). There has been recent interest in studying dynamic mechanism design in this setting--- see e.g., \cite{nisan2023serial, pai2024dynamic}. In a different context, a few market design papers have studied the design of queues, mainly in organ transplantation--see, e.g., \cite{leshno2022dynamic} and \cite{su2006recipient}.

\section{Model}\label{sec:model} 

Time is discrete, and each period corresponds to a point during which a validator may request a withdrawal and be removed from the active set --- for example, Ethereum processes withdrawals at epoch boundaries. Denote the set of possible validators, $V$, and at each time $t$, let $S(v,t) \in \{0,1\}$ denote whether validator $v$ is currently staked or not.\footnote{For simplicity, we assume that each validator has the same quantity of tokens staked, normalized to $1$. This can easily be relaxed.} Thus the total amount staked at period $t$ is $\overline{S}(t) =\sum_v S(v,t)$. 

At the end of each period, any validator may signal their desire to withdraw their stake by joining a waiting list $W(t)$. For now, we model every element of the waiting list as a tuple $(v,t')$, where $v$ is the validator identity and $t'$ is the period at which they initiated their withdraw. Note that we must have that $t' \leq t$, as any element in the waiting list must have joined in the past. Let $R(t)$ be the set of exit requests arriving in period $t$.\footnote{Most blockchains also limit entry to have a stable validator set for consensus. In this paper, we focus on the design of exit queues and consider entry unrestricted.}
 
\newcommand{\hist}{\mathcal{H}}
An exit mechanism $\mech$ in each period $t$, given a waiting list $W(t)$, selects a subset $P(t) \subseteq W(t)$ of withdrawal requests to \emph{process}. We allow the exit mechanism's choices to depend on past choices. Formally, let us define a history of previous withdrawal requests as:
\[H(t) = (P(1), P(2) \ldots, P(t-1)),\]
and the set of all possible histories as $\hist$. A mechanism then formally is:
\[\mech: \hist \times W(t) \mapsto \{0,1\}^{|W(t)|},\]
where the binary string is an indicator function for the withdrawals processed during each period.
The system follows the rules of motion:
\begin{align*}
    &W(t) = W(t-1) \setminus P(t-1)  \cup R(t),\\
    &P(t) = \mech( H(t), W(t)),\\
    &H(t+1) = H(t) \cup P(t). 
\end{align*}
The stake distribution $S(\cdot, t)$  is then updated based on the exits and fresh entries. In words, $W(t)$ is the \emph{waiting list}  of withdraw requests at the \emph{beginning} of period $t$, $P(t)$ is the subset of waiting to withdraw requests that are \emph{processed} in period $t$. $H(t)$ is the \emph{history} of processed requests up to and including period $t$. 

The number of withdrawals allowed over various time horizons constrains the protocol designer. We model this as a finite set of constraints, each described by a tuple $(\delta, T) \in [0,1] \times \mathbb{N}$. A constraint of $(\delta, T)$ means that if, in any period $t$, the total stake is $\overline{S}(t)$, then the maximum number of withdrawals processed over the following $T$ periods (from $t+1$ thru $t+T$) is bounded above by $\delta \times \overline{S}(t)$. We take the constraints as given, motivating this construction in Section \ref{sec:limiting}. Formally, the designer faces some $k$ constraints given by $\mathcal{C} = \{(\delta_1, T_1), \ldots, (\delta_k, T_k)\}$ and aims to maximize the utility of the validators withdrawing from the staking system. Calculating this utility depends on validators having differing values for exiting the system; we begin by examining the simplest case, where each validator has a common value.

\section{Homogeneous Values}\label{sec:homogeneous}
To begin our analysis, we consider the case where all agents have the same value for withdrawing or equivalently face the same economic penalty for each period between when they make a withdrawal request and when that request is fulfilled. In this case, the social planner cannot increase efficiency by reordering withdrawal requests, so efficiency demands that exit requests be processed as quickly as possible without violating the established constraints.

Given the constraint set $\mathcal{C}$, the following algorithm, which we call \texttt{MINSLACK}, greedily processes the maximum amount of withdrawals allowed within the bounds of the constraints. In other words, for every constraint $(\delta_i,T_i)$, calculate the difference between $\delta_i \overline{S}(t-T_i)$ and the amount withdrawn in periods $t-T_i$ thru $t-1$.%
\footnote{For expositional simplicity, we elide over the difficulties caused by the fact that $\delta_i \overline{S}(\cdot)$ may not be a whole number. In what follows we implicitly assume that this is a whole number, alternately, we could allow for fractional withdrawals at the cost of significantly messier notation.}
This difference is the maximum number of withdrawals constraint $i$ allows in period $t$ given the previous history. It follows that the lowest of these slacks is the maximum amount that can withdrawn in this period without violating constraints. Since the protocol is indifferent about the withdrawal order, if there is more demand for withdrawing than the allowed quantity, a natural solution is to use the FCFS rule to tie-break. We present this algorithm as Algorithm \ref{alg:1}.

\begin{algorithm}
\caption{\texttt{MINSLACK} \label{alg:1}}
\begin{algorithmic}[1]
\STATE \textbf{Input:} Constraints $\mathcal{C} = \{(\delta_1, T_1), \ldots, (\delta_k, T_k)\}.$
\STATE \textbf{Input:} Initial staking $S(\cdot,0)$. 
\STATE $\overline{S}(0) \gets \sum_v S(v,0).$
\STATE \textbf{Initialize:} $H(0), W(0), P(0) \gets$ \texttt{NULL}. 
\STATE \textbf{Initialize:} $\overline{P}(0) = 0$.
\FOR{each period $t \geq 1$}
\STATE $W(t) \gets W(t-1) \setminus P(t-1) \cup R(t).$
\FOR{each constraint $i \leq k$}
\STATE \texttt{SLACK}$_i \gets \delta_i \overline{S}(t-T_i) - \sum_{\tau = t-T_i+1}^{t-1} \overline{P}(\tau).$
\ENDFOR 
\STATE \texttt{MINSLACK} $\gets \min \{$\texttt{SLACK}$_i: 1 \leq i \leq k\}$.
\STATE $P(t) \gets$ Largest prefix of W(t) such that total withdrawn $\leq$ \texttt{MINSLACK}
\STATE $\overline{P}(t)$ $\gets$ Total withdrawn in $P(t)$
\STATE $H(t+1) \gets H(t) \cup P(t)$
\STATE \textbf{Update:} $S(v,t)$ based on $P(t)$.
\ENDFOR  
\end{algorithmic}
\end{algorithm}

\begin{figure}
    \centering
    \includegraphics[width=0.7\textwidth]{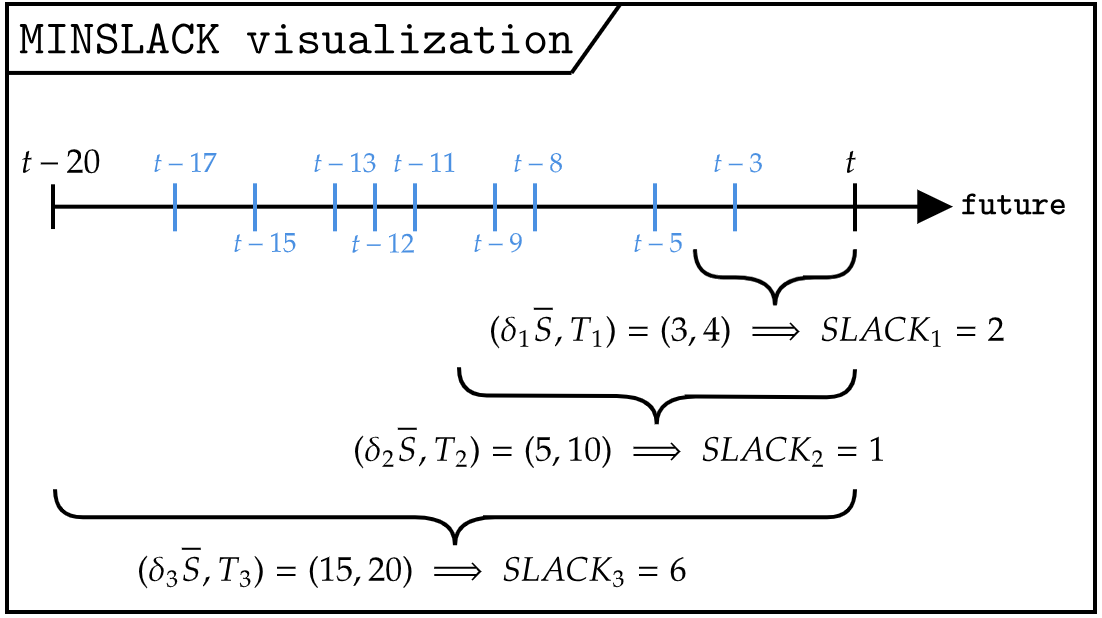}
    \caption{A visual example of the calculation of $\texttt{SLACK}_i$ used in Algorithm~\ref{alg:1} (\texttt{MINSLACK}). The example constraints $\mathcal{C} \implies \{(3,4), (5,10), (15,20)\}$ are read as, e.g., $(3,4) \implies $ ``at most three withdrawals over the next four consecutive time steps.'' In the diagrammed example, the blue vertical lines represent the timestamps of processed withdrawals. With $\texttt{SLACK}_2=1$, the \texttt{MINSLACK} algorithm can process at most one withdrawal during the current period while still conforming to the constraints.}
    \label{fig:minslackviz}
\end{figure}

Proving that this algorithm is feasible and optimal is straightforward: as designed, it processes the maximum amount allowed by the protocol constraints, but never more. Before presenting this result, we explain why such a queue design may be helpful. As we describe in Section~\ref{sec:eth}, the relevant constraints are that a given fraction of stake cannot withdraw over an extended period (e.g., $\mathcal{O}$(weeks)). Nevertheless, the actual queue implemented on Ethereum allows the withdrawal of at most eight validators per epoch (a value set in EIP-7514, \citealt{eip7514}). In practice, validators must wait longer than required during periods with higher-than-expected withdrawals. For example, in January 2024, the withdrawal queue on Ethereum rose to about 16,000 validators or about 5.5 days at peak due to Celsius bankruptcy proceedings.\footnote{See \url{https://www.validatorqueue.com/} for historical data about the withdrawal queue.} However, there were about 900k total validators during this period, so processing these withdrawal requests immediately would not have violated the consistency constraints defined by the ``weak-subjectivity period'' \citep{aditya2020wseth}. With this motivation, we present a formal treatment of \texttt{MINSLACK}.

\begin{theorem}\label{thm:homogenous}
    Given any sequence of withdrawal requests $R(\cdot)$, let $P(\cdot)$ be the processed withdrawal requests and $\overline{P}(\cdot)$ be the resulting total amount withdrawn in each period by Algorithm \ref{alg:1}. Then:
    \begin{enumerate}
        \item \textbf{Feasibility:} $P(\cdot)$ is feasible with respect to the protocol constraints. 
        \item \textbf{Optimality:} For any other feasible withdrawal decisions with total withdrawn in each period given by $\overline{P}'(\cdot)$, it must be the case that:
    \begin{align}\label{eqn:opt}
        \forall t \geq 1: & \sum_{\tau =1}^t \overline{P}'(\tau) \leq  \sum_{\tau=1}^t\overline{P}(\tau). 
    \end{align}
    \end{enumerate}
\end{theorem}
\begin{proof}
    To show that the withdrawal resulting from \texttt{MINSLACK} is feasible, observe that in each period, the withdrawal amount is less than $\min \{$\texttt{SLACK}$_i: 1 \leq i \leq k\}$ so it necessarily satisfies all of the constraints. Since each withdrawal satisfies the constraints given the history, applying the algorithm always results in a history that is feasible by construction.

    For optimality, consider for the sake of contradiction that there exists a feasible $\overline{P}'(\cdot)$ that violates \eqref{eqn:opt}. Let $t$ be the earliest time such that:
    \[\sum_{\tau =1}^t \overline{P}'(\tau) >  \sum_{\tau=1}^t\overline{P}(\tau). \]
    Since $t$ is the earliest time to violate condition \eqref{eqn:opt}, we must have that for all $t' <t$, 
    \begin{align}\label{eqn:claim}
      \forall t' < t: \quad   \sum_{\tau =1}^{t'} \overline{P}'(\tau) \leq   \sum_{\tau=1}^{t'}\overline{P}(\tau).
    \end{align}
    Analogous to the algorithm, let us term $\texttt{SLACK}'_i(\cdot)$ as the maximum withdrawable amount in a given period given process $P'(\cdot)$, with $\texttt{MINSLACK}'(\cdot)$ defined as the smallest constraint $i \in 1, \ldots k$. Note that by feasibility, we must have the following:
    \[\overline{P}'(t) \leq \texttt{MINSLACK}'(\cdot).\]
    Conversely, we know that by construction (see Algorithm \ref{alg:1}), 
    \[ \overline{P}(t) =\texttt{MINSLACK}(t).\]
    For the contradiction, it is sufficient to show that
    \begin{align} \label{eqn:nts}
    \texttt{MINSLACK}'(t) - \texttt{MINSLACK}(t) \leq \sum_{\tau=1}^{t-1} \overline{P}(\tau) - \sum_{\tau=1}^{t-1} \overline{P}'(\tau). 
    \end{align}
    In other words, we must show that the additional slack available at time $t$ under $P'$ relative to $P$ is, at most, the difference between the amount withdrawn up to time $t-1$ under $P$ than $P'$. Feasibility of the withdrawals under $P'$ then contradicts the claim that $\sum_{\tau =1}^{t'} \overline{P}'(\tau) \leq   \sum_{\tau=1}^{t'}\overline{P}(\tau)$.
    To see \eqref{eqn:nts} it is sufficient to show that for each $1 \leq k:$
    \begin{align} \label{eqn:nts2}
    \texttt{SLACK}_i'(t) - \texttt{SLACK}(t) \leq \sum_{\tau=1}^{t-1} \overline{P}(\tau) - \sum_{\tau=1}^{t-1} \overline{P}'(\tau). 
    \end{align}
    The left-hand side of \eqref{eqn:nts2} can be rewritten as:
    \begin{align*}
       \texttt{SLACK}_i'(t) - \texttt{SLACK}(t) &= \delta_i \overline{S}'(t-T_i) - \sum_{\tau = t-T_i+1}^{t-1} \overline{P}'(\tau) - \left(\delta_i \overline{S}(t-T_i) - \sum_{\tau = t-T_i+1}^{t-1} \overline{P}(\tau)\right)\\
       % =& -\delta_i \sum_{\tau=1}^{t-T_i}\overline{P}'(\tau) - \sum_{\tau = t-T_i+1}^{t-1} \overline{P}'(\tau) - \left(-\delta_i \sum_{\tau=1}^{t-T_i}\overline{P}(\tau) - \sum_{\tau = t-T_i+1}^{t-1} \overline{P}(\tau) \right)\\
       % =& \left(\delta_i \sum_{\tau=1}^{t-T_i}\overline{P}(\tau) + \sum_{\tau = t-T_i+1}^{t-1} \overline{P}(\tau) \right) - \left(\delta_i \sum_{\tau=1}^{t-T_i}\overline{P}'(\tau) - \sum_{\tau = t-T_i+1}^{t-1} \overline{P}'(\tau)\right)\\
       &\leq  \left( \sum_{\tau=1}^{t-T_i}\overline{P}(\tau) + \sum_{\tau = t-T_i+1}^{t-1} \overline{P}(\tau) \right) - \left( \sum_{\tau=1}^{t-T_i}\overline{P}'(\tau) - \sum_{\tau = t-T_i+1}^{t-1} \overline{P}'(\tau)\right)\\
       &=\sum_{\tau=1}^{t-1} \overline{P}(\tau) - \sum_{\tau=1}^{t-1} \overline{P}'(\tau).
    \end{align*}
where the penultimate inequality follows since $\delta_i \in [0,1]$, and we have \eqref{eqn:claim}.    
\end{proof}
Thus, \texttt{MINSLACK} is optimal for the common value setting. Still, in reality, stakers may have disparate values for accessing their stake, motivating the need to explore how a withdrawal mechanism could account for heterogeneous values.

\section{Heterogeneous Values and Paying for Priority}\label{sec:heterogenous}
While Theorem \ref{thm:homogenous} shows that Algorithm \ref{alg:1} provides an optimal solution for the case when all stakers have a homogeneous value for withdrawing, in reality, they may have different values for getting access to their staked assets. A staking pool, for example, might be withdrawing some validators gradually to rotate the cryptographic keys used for participating in consensus. In this case, the pool has a relatively low value for their withdrawal because the underlying reason to withdraw is not highly time-sensitive. On the other hand, a hedge fund trying to withdraw staked capital in time to meet a margin call to avoid a forced liquidation may have an extremely high value for the liquidity from the withdrawal processing.

Each validator looking to exit has a delay cost per unit time $c$. The net payoff of a validator of type $c$ whose withdrawal occurs after a delay $\Delta$ for a price (bid) of $b$ is:\footnote{This is a standard model in the context of transaction fees, where users face a similar trade-off between paying for inclusion and suffering a delay---see, e.g., \cite{huberman2021monopoly}.}
\[U(\Delta,b, c) = - c \Delta - b .\]
 In other words, their utility is linear in time according to their per-period disutility of waiting, less the amount they pay. We consider this canonical linear form for simplicity. More generally, one can consider other forms for the utility, including the time-varying disutility of waiting, see, e.g., \cite{bergemann2010dynamic}.

As described below, the efficient mechanism will be more complicated: efficiency requires agents to express their disutility of waiting in the mechanism and managing agents' incentives involve payments. As in the previous section, we will consider the planner's objective to be efficiency, which is defined formally below. 

\begin{observation}
   When values are heterogeneous, Algorithm \ref{alg:1}, \texttt{MINSLACK}, may not be efficient. 
\end{observation}
Recall that in every period, \texttt{MINSLACK} greedily processes as many withdrawal requests as possible, given the constraints. However, there are unknown future withdrawal requests at the time of processing. With heterogeneous values, it is possible that highly time-sensitive stakers with a high disutility of waiting may arrive in future periods. Suppose the current withdrawal requests have very low time sensitivity (i.e., very low $c$). In that case, the optimal behavior could be to withhold processing withdrawals in this period and reserve this capacity for the future. Intuitively, an efficient withdrawal mechanism must balance between processing withdrawals now while reserving some slack for hypothetical future withdrawals.   

\subsection{Efficient withdrawals under heterogeneity}
This section describes a withdrawal algorithm based on the Vickrey-Clarke-Groves (VCG) mechanism. VCG in such dynamic settings is not novel --- see \cite{parkes2003mdp} or \cite{lavi2000competitive}; it generalizes the second-price sealed-bid auction in static settings and has two desirable properties, namely, (i) it is incentive compatible for each agent to report their cost, $c$, and (ii) the mechanism is constrained-efficient. 

As is standard in mechanism design, we first describe the efficient allocation rule, i.e., the optimal rule for a planner, in a setting where \emph{the planner observes the delay costs of stakers as they arrive}. Then, we describe payment rules that make it \emph{incentive compatible} for stakers to report their values truthfully. 

Since this is a dynamic setting, as alluded to above, the mechanism must have a forecast of future arrivals to decide whether to process withdrawals or to reserve withdrawal slots for future arrivals. In this section, therefore, we assume that there is a known stochastic process behind the withdrawals. The number of withdrawal requests in each period is randomly distributed (for example, this may be the Poisson distribution with known parameter $\lambda$). Each withdrawal request has a type that is an i.i.d. draw according to a known probability distribution on $\Re_+$. 

\subsubsection{The Efficient Allocation Rule}
For now, suppose each agent truthfully states, at the time of joining the waiting list, their private cost, $c$. Each element of the waiting list is now a 4-tuple $(v,s,t',c)$. Modulo this change, however, the system can be described just as in Section \ref{sec:model}.

Given that some set of withdrawal requests $P(t)$ is processed in period $t$ from a waiting list of $W(t)$, the system collects a penalty (net disutility) of: 
\[ \texttt{Penalty} = \sum_{(v,s,t',c) \in W(t) \setminus P(t)} sc .\]
In other words, the planner in period $t$ collects a penalty equal to the disutility cost of every staker in the waiting list whose withdrawal is not processed. 
As before, the planner faces some constraints $\mathcal{C}$ on exits. The system aims to minimize expected discounted penalties over feasible exit plans, where $\rho \in [0,1]$ is the planner's discount rate.

This is a dynamic program where the state of the problem at the beginning of period $t$ is $(S(t), W(t), H(t-1))$. We can recursively define the value function of the planner as follows:
\begin{align} 
V(S(t), W(t), H(t-1)) &\equiv \label{prog:mdp}\\ \nonumber
 \min_{P(t)} &\bigg( \sum_{(v,s,t', c) \in W(t) \setminus P(t)} sc + \delta \mathbb{E}[V(S(t+1), W(t)  \setminus P(t) \\ \nonumber
 & \qquad \qquad \qquad \qquad \cup R(t+1),H(t) \cup P(t)]\bigg),\\ \nonumber
&\text{s.t. } P(t) \subseteq W(t), \\ \nonumber
    &\hphantom{\text{s.t. }} P(t) \text{ feasible wrt } \mathcal{C}. \nonumber
\end{align}
Here, expectations are taken over the next period withdrawal requests $R(t+1)$: both the number of withdrawal requests and the corresponding waiting disutility is unknown at period $t.$

This framing is an infinite horizon Markov Decision Problem (MDP). Given the previous history, there is a maximum number of feasible withdrawals in every period. For any withdrawal processed from the waiting list, it is intuitive that the planner will remove the ones with the highest disutility of waiting first. However, as described above, the marginal value of holding onto a withdrawal slot can exceed the penalty of making a current staker on the list wait an extra period. Of course, the precise details depend on the arrival process and the system's current state. The algorithm is described in Program~\ref{prog:mdp}.

\begin{algorithm}
\caption{\texttt{OPTIMAL} \label{alg:optimal}}
\begin{algorithmic}[1]
\STATE $\ldots$ \COMMENT{same as \texttt{MINSLACK}}
\FOR{each period $t \geq 1$}
\STATE $W(t) \gets W(t-1) \setminus P(t-1) \cup R(t).$
\STATE \textcolor{olive}{P(t) $\gets$ Solution of Program~\ref{prog:mdp}}
\STATE $\overline{P}(t)$ $\gets$ Total withdrawn in P(t)
\STATE H(t+1) $\gets$ H(t) $\cup$ P(t) 
\STATE \textbf{Update:} S(v,t) based on P(t), E(t).
\ENDFOR  
\end{algorithmic}
\end{algorithm}

\texttt{OPTIMAL} is nearly identical to \texttt{MINSLACK}; it only replaces the process for calculating the set of withdrawals to process in a current period, $P(t)$ (shown in brown text in Algorithm~\ref{alg:optimal}). The optimization problem in Program~\ref{prog:mdp} must be solved to determine the policy of how many withdrawals to process at each period.

\begin{algorithm}
\caption{\texttt{PRIO-MINSLACK} \label{alg:prio-minslack}}
\begin{algorithmic}[1]
\STATE $\ldots$ \COMMENT{same as \texttt{MINSLACK}}
\STATE \textcolor{blue}{Sort W(t) in decreasing order of waiting disutility.}
\FOR{each constraint $i \leq k$}
\STATE \texttt{SLACK}$_i \gets \delta_i \overline{S}(t-T_i) - \sum_{\tau = t-T_i+1}^{t-1} \overline{P}(\tau).$
\ENDFOR 
\STATE \texttt{MINSLACK} $\gets \min \{$SLACK$_i: 1 \leq i \leq k\}$.
\STATE $\ldots$ \COMMENT{same as \texttt{MINSLACK}}
\end{algorithmic}
\end{algorithm}

Another candidate withdrawal algorithm, which we refer to as \texttt{PRIO-MINSLACK}, modifies \texttt{MINSLACK} to process withdrawals in order of priority fees. Algorithm~\ref{alg:prio-minslack} represents this one-line change in blue text.

\subsubsection{Pricing Rule}
So far, we have described the problem as an optimization problem where the planner \emph{knows} the disutility from waiting suffered by the stakers in the queue. These are private, and there must be an incentive for stakers to report truthfully. Achieving this is straightforward (albeit computationally inefficient): every staker withdrawn in a period $t$ should pay the expected delay costs imposed on the system by their presence. Existing theorems (see \cite{parkes2003mdp}, \cite{bergemann2010dynamic}) show that such a pricing rule results in a Bayes-Nash equilibrium, where each buyer reports their values truthfully. 

\subsubsection{Optimal policy}\label{sub:optimal}

If new withdrawal arrival and value distributions are known, we can calculate the optimal withdrawal policy by solving the resulting MDP associated with Program~\ref{prog:mdp}.

\begin{figure}
    \centering
    \includegraphics[width=1\textwidth]{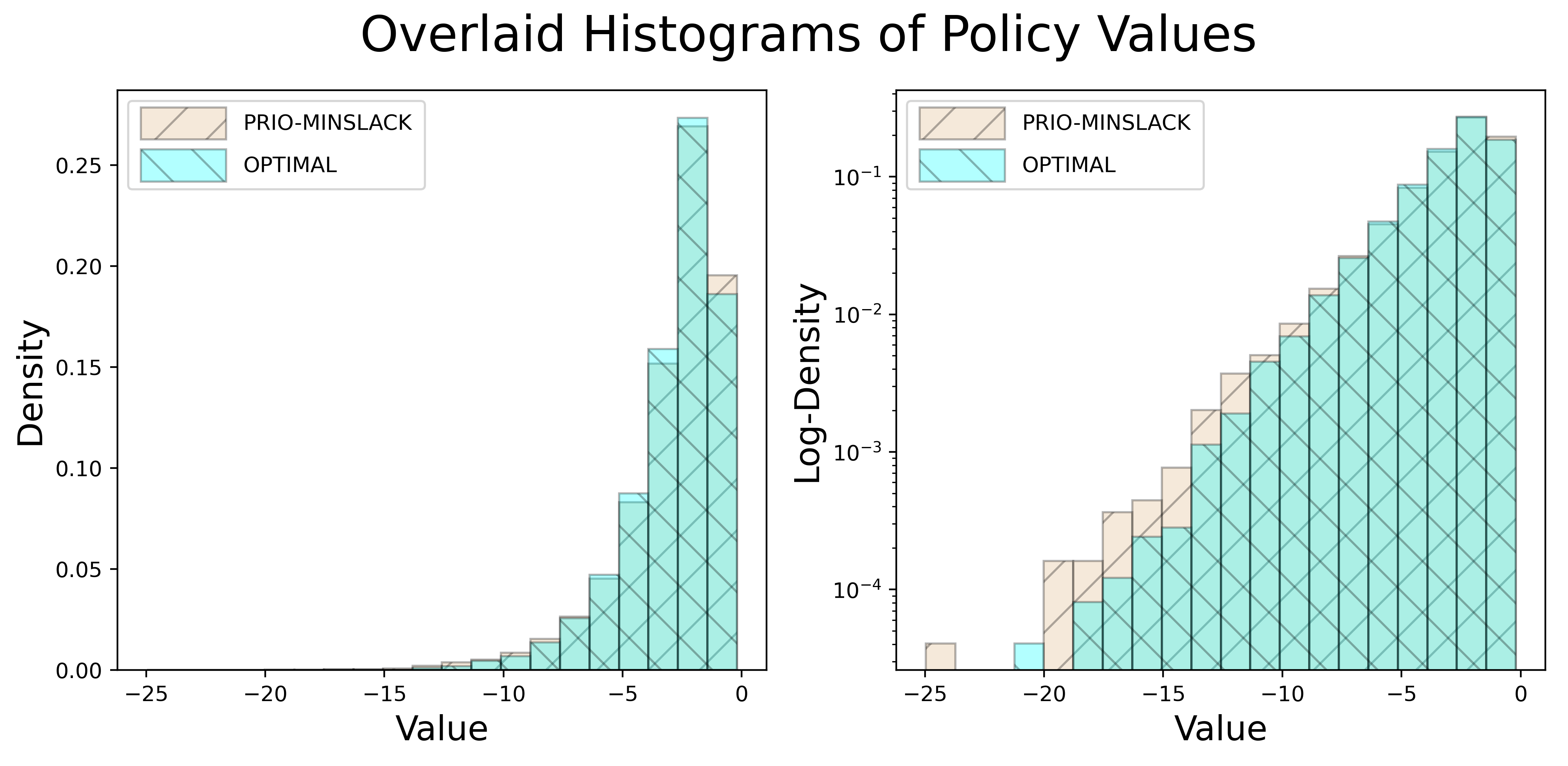}
    \caption{Performance comparison of \texttt{PRIO-MINSLACK} and \texttt{OPTIMAL} over 10,000 samples calculating the discounted reward following each policy from the initial state $s_0= [0,0,0,0,0,0]$ for 350 steps with a discount factor of $0.9$. The density of each histogram shows the probability a given trial ends in that range of values. When examining the raw density, the performance seems comparable, but the Log-Density plot demonstrates that the long tail performance of \texttt{PRIO-MINSLACK} is significantly worse than \texttt{OPTIMAL}. Intuitively, \texttt{PRIO-MINSLACK} is more of a ``gambler'' – the algorithm takes big risks by greedily processing as fast as possible. These risks are rewarded in the median case but occasionally have large disutilities by burning the capacity on low-value withdrawals. See Table~\ref{tab:optimalcomp} for more numerical comparisons between the two algorithms under different parameterizations.}
    \label{fig:histo}
\end{figure}

\begingroup
\setlength{\tabcolsep}{10pt}
\renewcommand{\arraystretch}{1.05}
\begin{table}
    \centering
    \begin{NiceTabular}{|c|c|c|c|c|c|}
    \hline & \\[-0.9em]
    \textbf{Algorithm} & \textbf{Arrival dist.} & \textbf{Value dist.} & \textbf{Discount} & \textbf{Performance}\\[0.2em]
    \hline
    \hline & \\[-0.9em]

    \texttt{OPTIMAL} &  & & \multirow{2}{*}{$0.85$}  & $-2.374$ \\
    \texttt{PRIO-MINSLACK} & & & & $-2.413$ \\
    
    \cmidrule{1-1}\cmidrule{4-5}
    
    \texttt{OPTIMAL} & $Y \sim [0, 1, 5]$ &  $X \sim [1, 10]$ & \multirow{2}{*}{$0.9$}  & $-2.933$ \\
    \texttt{PRIO-MINSLACK} &  $\textit{ w.p. } [0.5, 0.4, 0.1]$ & $\textit{ w.p. } [0.9, 0.1]$ & & $-2.982$ \\
    
    \cmidrule{1-1}\cmidrule{4-5}

    \texttt{OPTIMAL} &  &  & \multirow{2}{*}{$0.95$}  & $-3.964$ \\
    \texttt{PRIO-MINSLACK} & & & & $-3.999$ \\[0.2em]

    \hline  & \\[-0.9em]

    \texttt{OPTIMAL} & & $X \sim [1, 5]$ &  & $-2.428$ \\
    \texttt{PRIO-MINSLACK} & & $\textit{ w.p. } [0.9, 0.1]$ & & $-2.422$ \\
    
    \cmidrule{1-1}\cmidrule{3-3}\cmidrule{5-5}
    \texttt{OPTIMAL} & $Y \sim [0, 1, 5]$ & $X \sim [1, 10]$ & \multirow{2}{*}{$0.9$}  & $-2.959$ \\
    \texttt{PRIO-MINSLACK}  & $\textit{ w.p. } [0.5, 0.4, 0.1]$ & $\textit{ w.p. } [0.9, 0.1]$ & & $-3.005$\\
    \cmidrule{1-1}\cmidrule{3-3}\cmidrule{5-5}

    \texttt{OPTIMAL} &  & $X \sim [1, 20]$ &  & $-3.902$ \\
    \texttt{PRIO-MINSLACK} & & $\textit{ w.p. } [0.9, 0.1]$ & & $ -4.151$ \\[0.2em]

    \hline  & \\[-0.9em]
    
    \texttt{OPTIMAL} & $Y \sim [0, 1, 2]$ & &  & $-1.637$ \\
    \texttt{PRIO-MINSLACK} & $\textit{ w.p. } [0.4, 0.4, 0.2]$ & & & $-1.638$ \\
    
    \cmidrule{1-2}\cmidrule{5-5}

    \texttt{OPTIMAL} & $Y \sim [0, 1, 5]$ & $X \sim [1, 10]$ & \multirow{2}{*}{$0.9$}  & $-2.925$ \\
    \texttt{PRIO-MINSLACK} & $\textit{ w.p. } [0.5, 0.4, 0.1]$ & $\textit{ w.p. } [0.9, 0.1]$ & & $-2.969$ \\
    
    \cmidrule{1-2}\cmidrule{5-5}

    \texttt{OPTIMAL} & $Y \sim [0, 1, 10]$  & &  & $-3.610$ \\
    \texttt{PRIO-MINSLACK} & $\textit{ w.p. } [0.6, 0.35, 0.05]$ &  & & $-3.620$ \\[0.2em]

    \hline 
    \end{NiceTabular}
    
    \caption{Performance over 10,000 simulations for \texttt{OPTIMAL} and \texttt{PRIO-MINSLACK} under different configurations of arrival distributions ($Y$), value distributions ($X$), and discount factors. The performance metric is the discounted value of the rewards starting in the initial state $[0,0,0,0,0,0]$; higher values (smaller disutility) are better. \textbf{Top three pairs:} \textit{varying discount factors}. We use $n=225, 350, 700$ for simulation steps for the discount factors of $\gamma =0.85,0.90, 0.95$ respectively (each selected so that the end of the trial has a weighting of $\approx 10^{-16}$). \textbf{Middle three pairs:} \textit{varying the value distribution}. \textbf{Bottom three pairs:} \textit{varying the arrival distribution}.}
    \label{tab:optimalcomp}
\end{table}
\endgroup

\textit{A tractable instantiation.} Consider the withdrawal problem with a single constraint of $(t_0, \bar{S}\delta_0) = (5,5)$ (no more than five withdrawals are allowed over five time periods).%
\footnote{For our numerical exercises, for simplicity, we model the constraints as corresponding to an absolute number of validators that can withdraw over some window of periods.} 
Let the number of new withdrawals per period be distributed as $Y \sim \{0, 1, 5\}$ $\textit{ w.p. } \{0.5, 0.4, 0.1\}$ and the value of these distributions be distributed as $X \sim \{1, 10\} \\ \textit{ w.p. } \{0.9, 0.1\}$. We need only two values to represent the state of pending withdrawals, $W(t)$. Let $w_{\ell}$ and $w_{h}$ denote the number of pending ``low'' ($c=1$) and ``high'' ($c=10$) withdrawals, respectively. Further, let $h_{-1},h_{-2},h_{-3},h_{-4}$ denote the history of withdrawals processed (called $H(t-1)$ above) in each of the last four periods (with a $(5,5)$ constraint, this is the extent of the history that we must consider when deciding what withdrawals to process in this period). This leads the definition of each state $$s = [w_{\ell}, w_h, h_{-1},h_{-2},h_{-3},h_{-4}] \in S.$$ The action space in this MDP is $A = \{0,1,2,3,4,5\}$, where the action $a_i$ is legal if $\sum h_{j} + a_i \leq 5$. To limit the size of the action space, we only consider states where $w_\ell, w_h < 10$. Even with this extremely reduced setup, there are still $|A| \times |S| \times |S| = 6 \cdot 15246^2=  1394643096$ probabilities and rewards to encode. Nevertheless, this is feasible since the transition and reward matrices are sparse. 

Using value iteration, we numerically solve for the optimal policy, which determines, ``given a state, how many withdrawals should we process during this period.''  We now compare the performance of \texttt{OPTIMAL} (Algorithm~\ref{alg:optimal} (under an assumed discount factor of $0.9$)) and \texttt{PRIO-MINSLACK} (Algorithm~\ref{alg:prio-minslack}).\footnote{Table \ref{tab:optimalcomp} considers other discount factors, arrival processes, and distributions of value.} Recall that \texttt{PRIO-MINSLACK} is a much simpler heuristic, where it looks at the history and takes action $a_i = 5 - \sum h_j$. While this works well generally, there are situations where it is ``overly aggressive'' and can result in large disutitlities. For example, consider the state.
\begin{align*}
     [10, 0, 0, 0, 0, 0] \implies \text{10 pending lows, 0 pending highs, empty history}.
\end{align*}
In this situation, \texttt{PRIO-MINSLACK} observes that it can process five low withdrawals immediately and does so ($a_i=5$). The optimal policy, however, chooses $a_i=3$ instead. By processing five withdrawals in a single period, \texttt{PRIO-MINSLACK} forces a state where no more withdrawals are possible for the following four periods. Using the available capacity, the mechanism runs the risk of a high withdrawal arriving and needing to wait, resulting in a large disutility. The optimal algorithm is ``more cautious'' by reserving two withdrawal slots for the future, protecting for the possibility that a high-value withdrawal comes in the following few periods. Figure~\ref{fig:histo} shows a performance comparison of \texttt{PRIO-MINSLACK} and \texttt{OPTIMAL}.
There are ten states  (i.e., configurations of the current queue and history of withdrawals) in which the action dictated by the optimal policy differs from \texttt{PRIO-MINSLACK} by two (e.g., optimal processes two fewer withdrawals than \texttt{PRIO-MINSLACK}) and 338 states in which the optimal action differs from \texttt{PRIO-MINSLACK} by one. Table~\ref{tab:optimalcomp} compares the performance of \texttt{OPTIMAL} and \texttt{PRIO-MINSLACK} under a few variations of (i) arrival distributions, (ii) value distributions, and (iii) discount factors from simulating the two policies.

\subsection{Practical considerations for the heterogeneous value setting}

The previous section outlines dynamic VCG, the optimal withdrawal mechanism given known stationary arrival and value distributions. In practice, the social planner may not know these distributions, and further, the expected number of withdrawals or urgency of the demand for liquidity could change over time. Beyond this, implementing dynamic VCG would require solving the dynamic program outlined in Program~\ref{prog:mdp} and holding funds in escrow to execute the VCG payment rule---both of which seem possible on paper but present significant engineering challenges. 

\begin{algorithm}
\caption{$\alpha$-\texttt{MINSLACK} \label{alg:prop-minslack}}
\begin{algorithmic}[1]
\STATE $\ldots$ \COMMENT{same as \texttt{MINSLACK}}
\STATE \textcolor{blue}{Sort W(t) in decreasing order of waiting disutility.}
\FOR{each constraint $i \leq k$}
\STATE \texttt{SLACK}$_i \gets \delta_i \overline{S}(t-T_i) - \sum_{\tau = t-T_i+1}^{t-1} \overline{P}(\tau).$
\ENDFOR 
\STATE \texttt{MINSLACK} $\gets \min \{$\texttt{SLACK}$_i: 1 \leq i \leq k\}$.
\STATE \textcolor{red}{P(t) $\gets$ Largest prefix of W(t) such that total withdrawn $\leq \alpha\cdot$  \texttt{MINSLACK}}
\STATE $\ldots$ \COMMENT{same as \texttt{MINSLACK}}
\end{algorithmic}
\end{algorithm}

\begingroup
\setlength{\tabcolsep}{10pt}
\renewcommand{\arraystretch}{1.05}
\begin{table}
    \centering
    \begin{NiceTabular}{|c|c|c|c|}
    \hline & \\[-0.9em]
    \textbf{Algorithm} & \textbf{Arrival dist.} & \textbf{Value dist.} & \textbf{Performance}\\[0.2em]
    \hline
    \hline & \\[-0.9em]
    \texttt{CONSTANT} (1) & & \multirow{4}{*}{$X \sim \text{Uniform}(0,1)$} & $-5.768$ \\
    \texttt{MINSLACK} & & & $-5.464$ \\
    \texttt{PRIO-MINSLACK} & & & $-2.019 $ \\
    \texttt{$\alpha$-MINSLACK} ($\alpha = 0.9$) & & & $-2.002 $ \\
    
    \cmidrule{1-1}\cmidrule{3-4}

    \texttt{CONSTANT} (1) & \multirow{3}{*}{$Y \sim [0, 1, 5]$} &  \multirow{4}{*}{$X \sim \text{Exp}(0.1)$} & $-12.249$ \\
    \texttt{MINSLACK} & \multirow{3}{*}{$\textit{w.p. } [0.5, 0.4, 0.1]$} &  & $-11.648$ \\
    \texttt{PRIO-MINSLACK} &  & & $-2.951$ \\
    \texttt{$\alpha$-MINSLACK} ($\alpha = 0.9$) & & & $-2.986$ \\
    
    \cmidrule{1-1}\cmidrule{3-4}

    \texttt{CONSTANT} (1) &  & \multirow{4}{*}{$X \sim \text{Pareto}(2,5)$} & $-114.913$ \\
    \texttt{MINSLACK} &  & & $-109.354$ \\
    \texttt{PRIO-MINSLACK} & & & $-67.687$ \\
    \texttt{$\alpha$-MINSLACK} ($\alpha = 0.9$) & & & $-63.070$ \\[0.2em]
    
    \hline
    \end{NiceTabular}
    
    \caption{Numerical results for algorithm performance under a fixed withdrawal arrival distribution and three different value distributions. The performance metric measures the average disutility over the withdrawals and thus should be minimized (to maximize the utility). We calculate the mean disutility over ten independent samples of 10,000 steps each, with the first 1,000 steps of each sample discarded to allow the system to settle into a steady state. The single constraint was set as $(\delta, T)= (5,5)$: ``a maximum of five withdrawals may be processed over five periods.'' \texttt{CONSTANT} processes one withdrawal per time step. \texttt{MINSLACK} and \texttt{PRIO-MINSLACK} follow the descriptions in Algorithms~\ref{alg:1} and \ref{alg:prio-minslack} respectively. For \texttt{$\alpha$-MINSLACK} (Algorithm~\ref{alg:prop-minslack}), we use $\alpha=0.9$, which produces the following mapping for calculating how much slack to consume in a given time slot $[0,1,2,3,4,5] \mapsto [0,1,2,3,4,4]$. We can describe this simply as: ``If the slack is exactly 5, use only four (reserving one for a potentially high-value arrival). If the slack is less than 5, use it entirely.'' The arrival distribution mimics occasional bursts of withdrawal requests while maintaining an expected value $E[Y] =0.9$, less than the average capacity of one derived by the $(5,5)$ constraint. The withdrawal values were sampled from Uniform, Exponential, and Pareto distributions to demonstrate that under some conditions, \texttt{$\alpha$-MINSLACK} can outperform \texttt{PRIO-MINSLACK}. In all cases, \texttt{CONSTANT} and \texttt{MINSLACK} perform far worse than the \texttt{PRIO-} and \texttt{$\alpha$-} variants.}
    \label{tab:numerical}
\end{table}
\endgroup

\texttt{PRIO-MINSLACK} is much simpler to implement, but may suffer under value heterogeneity because it is too eager to process withdrawals. The problem arises when \texttt{PRIO-MINSLACK} receives a burst of low-value withdrawals, in which case it consumes all the available capacity on the low-priority withdrawals and leaves potentially higher-value incoming withdrawals pending longer. These bursts induce a natural question: can we modify \texttt{PRIO-MINSLACK} to be slightly more conservative with its remaining capacity while preserving its simplicity? One solution is to modify \texttt{PRIO-MINSLACK} to consume only an $\alpha \in (0,1]$ proportion\footnote{The interval is left-open because $\alpha=0$ implies no withdrawals are ever processed.} of the \texttt{SLACK} available at each period.

Algorithm~\ref{alg:prop-minslack}, which we call \texttt{$\alpha$-MINSLACK}, makes the one-line modification (shown in red) to \texttt{PRIO-MINSLACK} by scaling the amount of processed withdrawals by $\alpha$. By tuning $\alpha$, we can make \texttt{$\alpha$-MINSLACK} more or less aggressive in how much withdrawal capacity it uses now versus saving.
At $\alpha=1$, we reduce to the ``maximally aggressive'' version (\texttt{PRIO-MINSLACK}). In contrast, as  $\alpha\to 0$, \texttt{$\alpha$-MINSLACK} becomes increasingly conservative. 
The outcome here is that the slack continues to build up across the constraints, and you end up processing at the rate $\alpha \cdot \delta_i / T_i$ per-unit time, where $(\delta_i, T_i) \in \mathcal{C} \text{ s.t., } \delta_i/T_i = \min_j \delta_j/T_j$. In other words, process at a constant rate proportional to the ``most restrictive'' constraint in $\mathcal{C}$. More moderate values, e.g.,  $\alpha = 0.5$, present a more balanced version of \texttt{$\alpha$-MINSLACK} where the algorithm functions on the heuristic of ``using half of the remaining slack at each period.''

While \texttt{$\alpha$-MINSLACK} is not necessarily optimal, it does eliminate the need for arrival distribution knowledge, which the optimal mechanisms rely on. Further, its simplicity makes it much more feasible for an actual production system. We justify this statement by numerically comparing the performance of various mechanisms under different withdrawal distributions. Note that we can expand the set of value distributions compared to the optimal analysis because we are no longer constructing the entire state space of the MDP.

Table~\ref{tab:numerical} compares the performance of four different algorithms across a constant arrival distribution and under three different value distributions. These results demonstrate that the \texttt{PRIO-} and \texttt{$\alpha$-} versions of \texttt{MINSLACK} far outperform either the \texttt{CONSTANT} mechanism or regular \texttt{MINSLACK} (which serves as an FCFS-queue rather than a priority queue based on the value of the withdrawal). These results motivate that, under some distributions, \texttt{$\alpha$-MINSLACK} may be preferable to \texttt{PRIO-MINSLACK}. Further work could be done to study adaptive algorithms that aim to learn the optimal value of $\alpha$ in an on-line fashion. Again, these heuristic rules for determining the withdrawal policy of the staking system are far more straightforward to construct and implement than the optimal versions described in Section~\ref{sub:optimal}.

\section{Theory and Practice}\label{sec:theory-and-practice}\label{sec:limit}
\noindent\textbf{Why limit withdrawals in the first place? A thought experiment.} Assume that withdrawals are not limited. An attacker, Eve, accumulates $1/3$ of the total stake in the PoS mechanism and invests heavily in networking infrastructure. Eve contacts Alice to inquire about buying a Tesla Cybertruck©. Alice, who is feeling both cyber- and cypherpunk enough to accept \texttt{ETH} for the transaction, sees \texttt{txn 0xcb} on Etherscan as finalized, giving her confidence to hand the (car) keys to Eve. 
From Alice's perspective, the settlement assurance of $1/3$ of all staked \texttt{ETH} ($>33$ billion USD as of May 2024) is more than sufficient economic security for her transaction. However, Eve using her networking prowess, Eve had tricked the honest validators into finalizing two conflicting blocks, one which included \texttt{txn 0xcb} and another that didn't by partitioning the honest validators into two separate p2p groups and sharing conflicting attestations with each group. If withdrawals are not limited, she can fully withdraw her stake from both chains by the time honest validators reconnect (once Eve's network-level attack ends) and try to slash her. Alice has no Telsa Cybertruck© nor the \texttt{ETH} originally sent in \texttt{txn 0xcb}.

In light of this, blockchains place limits on withdrawals. However, as described below, there is substantial variation in the limits placed and the withdrawal procedure, with little systematic study. 

% Staking mechanisms provide ``economic security'' to the applications they support through the credible threat of slashing the collateral provided to the system. These mechanisms have increased throughout the blockchain ecosystem.    

\subsection{Accountable Safety and Limiting Withdrawals}\label{sec:limiting}
We begin with the following simple observation. 
\begin{observation}
    The accountable safety of a finalized block \textbf{decreases} as time passes because the stake participating in the finalization of the block can withdraw from the system.
\end{observation}
\noindent This (rather counter-intuitive) fact means protocol designers must decide: ``How quickly should validators be able to withdraw their stake from the system?'' Let $\mathcal{D}$ denote the `maximum-tolerable decay' in the accountable safety of a finalized block. For example, if $\mathcal{D} = 1/6$, then a finalized block may have accountable safety (in terms of proportion of the total stake that is slashable in case the transaction history changes) of $1/3-\mathcal{D} = 1/6$. The security decay modifies the statement to, ``any transaction in a finalized block will have accountable safety of at least $1/6$ of all stake.'' This remains incomplete because over a sufficiently long time horizon, with withdrawals enabled, more than $\mathcal{D}$ stake may be removed from the system. 
Thus, we define an amount of time, denoted $\delta$, over which the stake withdrawn must not exceed $\mathcal{D}$. This period can serve multiple purposes. One such usage is the weak-subjectivity period \citep{buterin2014howilearned}, where the delay is an upper bound on the communication delay between all honest parties in the partially-synchronous protocol; this value is $\mathcal{O}(weeks)$ to account for the natural overhead incurred when social coordination is required to come to consensus.\footnote{By ignoring any blocks published prior to the weak-subjectivity checkpoint, validators can also eliminate the risk of long-range attacks (in practice, validators treat their latest finalized block as a `genesis` or irreversible block by simply rejecting any block that conflicts with it).} Other constraints might be over much shorter time horizons, e.g., $\mathcal{O}(minutes)$, to ensure a bound on the rate at which the economic security of a block changes in short windows. 
Thus, the accountable safety of a Proof-of-Stake mechanism parameterized by $\mathcal{D}$ and $\delta$ is ``any block finalized more than $\delta$ time ago is immutable (only social consensus could reverse it), and any block finalized within the past $\delta$ time has accountable safety of at least $1/3-\mathcal{D}$.''

\subsection{Ethereum}\label{sec:eth} Withdrawals in Ethereum Proof-of-Stake were fully activated in the Shanghai/Capella Hardfork\footnote{\url{https://ethereum.org/en/history/\#shapella}} on April 12, 2023. While the full withdrawal process is quite involved, we dig into the details to demonstrate how much engineering can shape the withdrawal mechanisms in use today. Figure~\ref{fig:ethereum-withdrawals} demonstrates the full flow of an Ethereum withdrawal, which is split into three distinct phases.

\begin{figure}
    \centering
    \includegraphics[width=\textwidth]{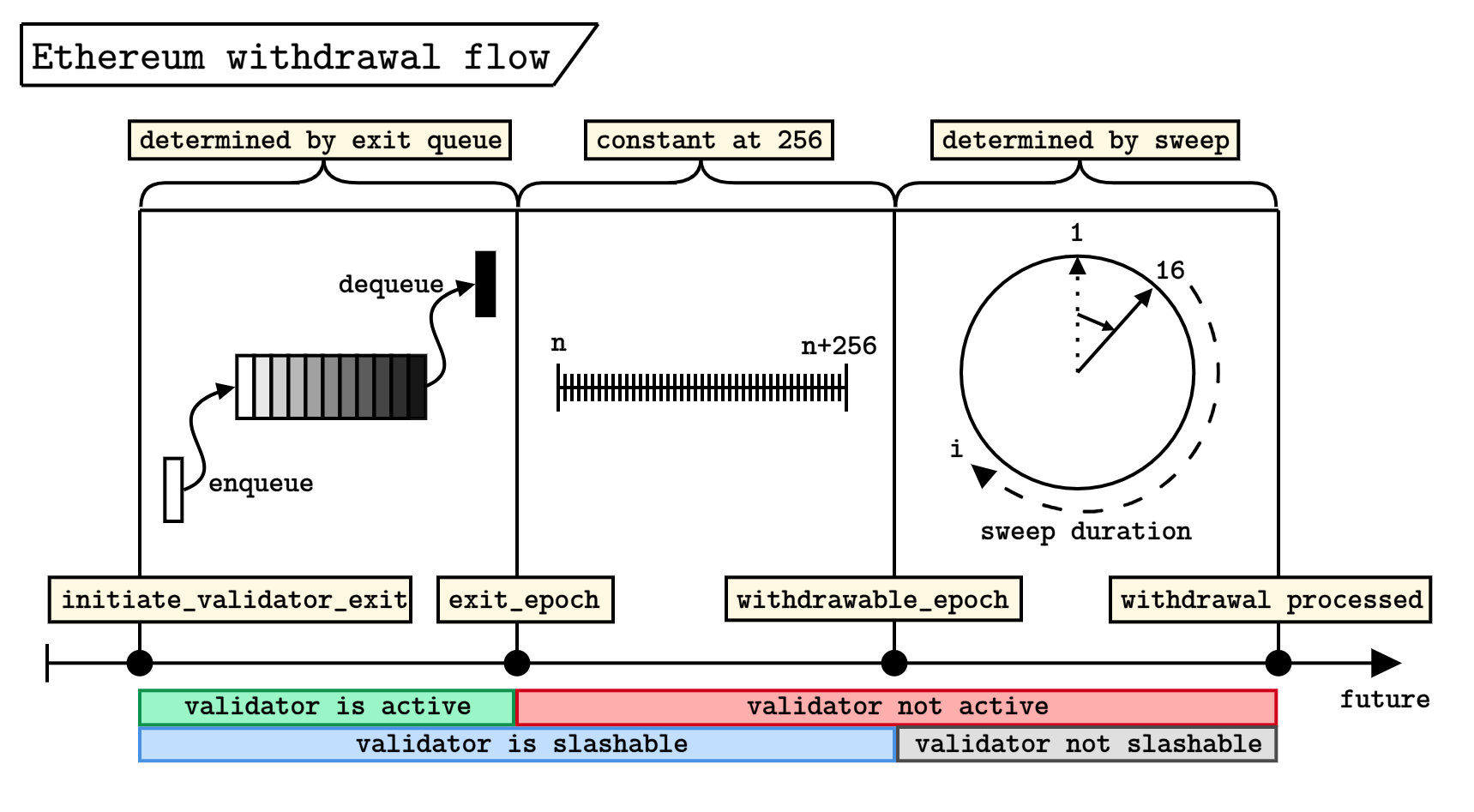}
    \caption{The withdrawal flow for Ethereum validators. Each phase has differing lengths and validator properties. The \textbf{top row} of tan labels demonstrate what determines the length of each phase. The \textbf{middle row} of tan labels annotate the timeline of events as described in the \cite{capellaspec}. The \textbf{bottom row} of colored labels indicate the activity and slashability of the validator over time.}
    \label{fig:ethereum-withdrawals}
\end{figure}

\textit{Phase 1: Exit queue –} When an Ethereum validator wants to withdraw their 32 \texttt{ETH} from the consensus mechanism, they trigger a ``Voluntary Exit'' \citep{phase0spec}. This process sets the validator's \texttt{exit\_epoch} based on the rate-limited first-come-first-served exit queue; during each epoch, at most $\min(4, \lfloor \# \text{ validators} / 2^{16} \rfloor$) are processed \citep{phase0spec} (though this was changed in EIP-7514 to cap the churn limit at $8$ validators per-epoch \cite{eip7514}, making the new function $\max(8, \min(4, \lfloor \# \text{ validators} / 2^{16} \rfloor))$). The \texttt{CHURN\_LIMIT\_QUOTIENT} $=2^{16}$ was selected\footnote{Powers of two common for specification constants due to their compact binary representation.} according to the rough heuristic that it should take approximately one month for $10\%$ of the stake to exit (or equivalently, about 100 days for $33\%$ of the stake to exit) \cite{annotated}.
For the entire time a validator is in the exit queue (this phase), they are both `active' (meaning they must continue performing their consensus duties) and `slashable' (meaning their stake is still accountable for their behavior). Keeping the validator active while in the exit queue minimizes the economic cost of a very long exit queue because they continue earning rewards \citep{buterin2019entriesnotwithdrawals}.  

\textit{Phase 2: Withdrawalability delay –} Once the validator's \texttt{exit\_epoch} has passed, they incur a constant delay of 256 epochs (27 hours) before their \texttt{withdrawable\_epoch} \citep{phase0spec}. This fixed delay is a significant safety buffer to provide ample time for the protocol to include any slashing proof on chain. During this time, the validator is no longer active (and thus not earning any rewards), but they remain slashable (to avoid committing a slashing violation immediately before the withdrawal). The enforcement of this delay ensures that, even if the exit queue is empty, there is a period where the validator's stake is still accountable for their actions.

\textit{Phase 3: Validator sweep –} Once past the validator's \texttt{withdrawable\_epoch}, the function \texttt{is\_fully\_withdrawable\_validator} returns \texttt{true} indicating that the withdrawal delay has passed and the validator is no longer slashable \citep{capellaspec}. The last delay comes from the amount of time it takes for the actual withdrawal requests to send the \texttt{ETH} to the corresponding withdrawal address. All withdrawals are processed by looping through the validator set in order of validator index. This validator ``sweep'' can only process 16 withdrawals per block, corresponding to 8.8 days to iterate through the entire validator set (thus a 4.4-day additional delay on average).\footnote{\url{https://www.validatorqueue.com/}} This 8.8-day delay is present regardless of the length of the exit queues because ``full'' withdrawals (where a validator wants to leave the consensus layer altogether) are inter-mixed with ``partial'' withdrawals (where a small amount of validator rewards are involuntarily swept from each validator). While the original specification implemented the queues directly into the protocol, \cite{potuzpr} changes this only to store the validator index and perform the sweep by mixing partial and full withdrawals. 

This withdrawal mechanism is quite complex; the minimum time to exit the system is 27 hours. Due to the validator sweep, if the validator doesn't strategically time their exit request, the withdrawal will take over 5.5 days on average to fully clear, even if the exit queue is empty. This complexity highlights how engineering decisions can inform the exit queue mechanism design. Beyond Ethereum, there are many additional staking systems, though their withdrawal mechanisms are much simpler and thus presented in Table~\ref{tab:proof-of-stake-table} \& Table~\ref{tab:applayer-table}.

\begingroup
\setlength{\tabcolsep}{10pt}
\renewcommand{\arraystretch}{1.2}
\begin{table}
    \centering
    \begin{tabularx}{\linewidth}{|c||c|L|L|}
    \hline
    \textbf{Protocol} & \textbf{Staking purpose} & \textbf{Withdrawal mechanism} & \textbf{One-line analysis}\\
    \hline
    \hline
    \textit{Ethereum} & Consensus safety & Rate-limited FCFS queue with minimum duration. & Aims to be fast in the average case, but partial withdrawals induce high-variance delay.\\
    \hline
    \textit{Cosmos} & Consensus safety & Fixed 21-day unbonding period. & Simple but inefficient.\\
    \hline
    \textit{Solana} & Sybil resistance & All deactivations happen at epoch boundaries. A maximum of $25\%$ of stake can deactivate at any given epoch boundary. & With no slashing, stake does not provide accountable safety to the protocol. Limiting withdrawals ensures the entire stake cannot exit in a single epoch. \\
    \hline
    \textit{Cardano} & Sybil resistance  & Immediate withdrawals. & With no slashing, stake does not provide accountable safety to the protocol. Withdrawals are immediately processed. \\
    \hline
    \textit{Polygon} & Consensus safety & Fixed $\approx 40$ hour unbonding period. &  Simple but inefficient. It benefits from the fact that, as a sidechain, state updates are posted to Ethereum and are thus immutable – allowing for a relatively shorter fixed duration.\\
    \hline
    \textit{Polkadot} & Consensus safety & Fixed 28-day unbonding period. & Simple but inefficient. \\
    \hline
    \textit{Avalanche} & Sybil resistance & Validators dictate the duration of their staking before becoming active. The minimum duration is two weeks. After time has elapsed, the stake is immediately withdrawn. & With no slashing, stake does not provide accountable safety to the protocol. Withdrawals are immediately processed.  \\
    \hline
    \end{tabularx}
    \caption{Comparing staking and withdrawal mechanisms across L1 protocols and sidechains.}
    \label{tab:proof-of-stake-table}
\end{table}
\endgroup

\subsection{Other Proof-of-Stake Blockchains}
Table~\ref{tab:proof-of-stake-table} compares several other blockchain protocols and how they handle withdrawals. Ethereum is the only protocol that implements a dynamic queue, and in this regard, Ethereum takes on additional complexity to improve the efficiency of the withdrawal mechanism. Cosmos, Polygon, and Polkadot each implement the simple, fixed-duration withdrawal mechanism with delays of 21, 2, and 28 days, respectively. This mechanism is simple and easy to reason about. Still, it is much less efficient because each withdrawal takes the maximal amount of time regardless of the history of the mechanism \citep{cosmosunbonding, polygon, polkadocs}. Solana, Cardano, and Avalanche do not have in-protocol slashing, so staking serves only as an anti-Sybil mechanism in their systems; the stake can exit the system without a rate-limiting step and not change their security model \citep{solanadeactivate,cardanoundelegation, avalanche}.

\begingroup
\setlength{\tabcolsep}{10pt}
\renewcommand{\arraystretch}{1.3}
\begin{table}
    \centering
    \begin{tabularx}{\linewidth}{|c||L|L|L|}
    \hline
    \textbf{Protocol} & \textbf{Staking purpose} & \textbf{Withdrawal mechanism} & \textbf{One-line analysis}\\
    \hline
    \hline
    \textit{EigenLayer} & Economic security guarantees & Fixed 7-day escrow period for all \texttt{ETH}-denominated withdrawals. Staked \texttt{EIGEN} has a fixed 24-day escrow period. & Withdrawals need to be limited because EigenLayer introduces new slashing conditions. Native restaked \texttt{ETH} may be withdrawn from the beacon chain during the EigenLayer escrow period. Each AVS could add its rate limiting in addition to the system-wide minimums. \\
    \hline
    \textit{Chainlink} & Oracle safety & Fixed 28-day cool-down period before \texttt{LINK} is claimable. & Staking provides safety and availability conditions for data feeds. Withdrawals are rate-limited to ensure slashing has time to take place. \\
    \hline
    \end{tabularx}
    \caption{Comparing staking and withdrawal mechanisms between EigenLayer and Chainlink, two app-layer protocols with slashing.}
    \label{tab:applayer-table}
\end{table}
\endgroup

\subsection{Other applications of staking}
Beyond other blockchains, some applications have implemented staking and slashing mechanisms at the application layer of Ethereum. Table~\ref{tab:applayer-table} performs the same high-level analysis of two such mechanisms.

EigenLayer and Chainlink use stake for slightly different purposes than the chains outlined in Table~\ref{tab:proof-of-stake-table}. EigenLayer creates a platform for buying and selling ``economic security''; services built on EigenLayer (called ``Actively Validated Services'' or `AVSs') purchase this security by incentivizing capital to delegate to an operator running their service. Because EigenLayer encumbers capital with additional slashing conditions, it also enforces a protocol-wide escrow period for stake removal. It is worth noting that the Ethereum withdrawal period can occur concurrently with the EigenLayer escrow period \citep{eigelayerescrow}. 
Further, services buying security from EigenLayer can impose further constraints on the capital allocated to their system. 
Chainlink, on the other hand, uses stake to provide security for the data feeds supplied by their oracle network \citep{breidenbach2021chainlink}. This stake may be slashed for ``less objective'' faults (e.g., slashing for being offline and not providing a price feed), which was recently dubbed ``inter-subjective slashing'' in \cite{eigelayerpaper} and may grow to play a significant role in the future designs of slashing protocols.

\subsection{Liquid staking \& restaking tokens}
Liquid staking tokens (LSTs) make a design trade-off when choosing how much of the capital in their system to deploy into consensus mechanisms. If they deploy too much of it, the withdrawals will be rate-limited by the underlying protocol, leading to a more capital-efficient protocol at the cost of a worse UX (slower withdrawals). Keeping some liquidity available for immediate redemption improves the UX, but any capital in that state is not cash-flowing. LSTs are fully collateralized and thus do not face insolvency risk, but holders face the duration risk of holding the LST for however long the withdrawal takes. Liquid restaking tokens (LRTs) have a more complex design space, where they must balance withdrawals against various underlying protocols and services. Their withdrawal mechanisms are plagued by the nature of various protocol rewards denominated in different tokens and emissions rates. \cite{lrtsfromfirst} explores some design trade-offs, including a market for withdrawals. Overall, this design space is extensive and out-of-scope for the modeling of this paper, but it presents an exciting avenue for future research. 

\section{Conclusion}\label{sec:concl}

System designers of staking and restaking protocols face a fundamental trade-off between the security and utility. Based on the mechanisms we surveyed in Section \ref{sec:theory-and-practice}, the mechanisms currently in production maximally flexible given the rigidity they claim to require. In other words, nobody seems to be on the production-possibilities frontier of egress mechanisms in practice; we acknowledge that the practical engineering constraints, e.g., as described in the design of Ethereum's withdrawal mechanism in Section~\ref{sec:eth}, may play a significant role in the decision making of existing protocols. 

By formalizing this trade-off as a constrained optimization problem over mechanisms, we aim to improve the state of withdrawal systems more broadly. For blockchain designers, we distill our results into three pieces of advice. First, suppose your consistency constraints are over a longer time horizon than a single epoch. In that case, a queue with dynamic capacity can significantly reduce average wait times without sacrificing security --- \texttt{MINSLACK} (Algorithm~\ref{alg:1}) is a simple example of maximally processing the rate of withdrawals given a set of constraints. 
Second, if you believe that participants in the system may have heterogeneous disutility from waiting in the exit queue, their welfare would be improved by implementing a priority queue -- \texttt{PRIO-MINSLACK} (Algorithm~\ref{alg:prio-minslack}) can quickly decrease the overall disutility. 
Third, if you think that the time-sensitivity or arrival process of future withdrawal requests is particularly fat-tailed, be sure to reserve some capacity in the system to allow the processing of highly time-sensitive withdrawals during periods of congestion --- \texttt{$\alpha$-MINSLACK} (Algorithm~\ref{alg:prop-minslack}) is an example of this reservation. 

We point to a few intriguing directions in terms of future work. Firstly, several empirical questions have been raised by this study. Assessing the actual staker surplus lost from sub-optimal queue designs would be helpful. The protocol may care about this staker surplus because reducing the staker disutility may lessen the emissions needed to incentivize token holders to stake in the first place. Further study on the heterogeneity in time preferences among stakers would help determine whether pay-for-priority systems are worth considering. Lastly, validator utility functions that are non-linear (e.g., a validator who needs their withdrawal within the next week but doesn't care when) may lead to different design considerations and optimal withdrawal mechanisms.

On the theoretical side, note that some of the pay-for-priority systems we have proposed serve as benchmarks and are unlikely to be implementable in practice (e.g., the dynamic programming-based efficient allocation in Algorithm \ref{alg:optimal}, which is both computationally difficult and requires knowledge of the distribution of withdrawal requests). Other mechanisms (e.g., \texttt{PRIO-MINSLACK}, Algorithm \ref{alg:prio-minslack}) are feasible as a pay-your-bid mechanism --- reminiscent of the Bitcoin (and Ethereum before EIP-1559) transaction-fee mechanisms. Similar concerns faced in those contexts, users having to choose an appropriate bid, may apply in withdrawal mechansims too. The natural question is whether designs with better user experience, analogous to EIP-1559, exist in this setting.

\newpage

\bibliographystyle{econometrica}
\bibliography{tfm}

\end{document}